\documentclass{article}

\usepackage{latexsym}
\usepackage{amssymb}
\usepackage{amsmath}
\usepackage{amsthm}
\usepackage{booktabs}
\usepackage{enumitem}
\usepackage{graphicx}
\usepackage{color}

\usepackage{soul}
\usepackage{url}
\usepackage{algorithm}
\usepackage{algorithmic}
\usepackage{natbib}
\usepackage{nicefrac}

\newtheorem{theorem}{Theorem}
\newtheorem{observation}[theorem]{Observation}
\newtheorem{problem}[theorem]{Problem}
\newtheorem{example}[theorem]{Example}
\newtheorem{definition}[theorem]{Definition}
\newtheorem{lemma}[theorem]{Lemma}
\newtheorem{claim}[theorem]{Claim}

\newcommand{\p}{{\rm P}}
\newcommand{\fp}{{\rm FP}}
\newcommand{\fl}{{\rm FL}}
\newcommand{\np}{{\rm NP}}

\newcommand{\optl}{{\rm OptL}}
\newcommand{\nc}{{\rm NC}}
\newcommand{\fnc}{{\rm FNC}}
\newcommand{\nl}{{\rm NL}}

\newcommand{\BibTeX}{B\kern-.05em{\sc i\kern-.025em b}\kern-.08em\TeX}

\begin{document}

\title{On the Parallelizability of Approval-Based Committee Rules}
\date{August 4, 2025}
\author{Zack Fitzsimmons\\
 Dept.\ of Math.\ and Computer Science\\
 College of the Holy Cross\\
 Worcester, MA 01610 \and
 Zohair Raza Hassan\\
 Department of Computer Science\\
  Rochester Institute of Technology \\
  Rochester, NY 14623 \and
  Edith Hemaspaandra\\
  Department of Computer Science\\
  Rochester Institute of Technology \\
  Rochester, NY 14623}%

\maketitle

\begin{abstract}
Approval-Based Committee (ABC) rules are an important tool for choosing a fair set of candidates when given the preferences of a collection of voters. Though finding a winning committee for many ABC rules is \np-hard, natural variations for these rules with polynomial-time algorithms exist. The recently introduced Method of Equal Shares, an important ABC rule with desirable properties, is also computable in polynomial time. However, when working with very large elections, polynomial time is not enough and parallelization may be necessary. We show that computing a winning committee using these polynomial-time ABC rules (including the Method of Equal Shares) is \p-hard, thus showing they cannot be parallelized. In contrast, we show that finding a winning committee can be parallelized when the votes are single-peaked or single-crossing for the important ABC rule Chamberlin-Courant.
\end{abstract}

\section{Introduction and Related Work}

Elections are a widely-used tool for a group of agents to reach a collective
decision. Multiwinner voting rules are used for elections where instead of a single winner, the desired outcome is a set of winners of a given size (see~\citet{fal-sko-sli-tal:b:trends-multiwinner}).
Approval-based committee (ABC) rules are 
an important and widely-studied type of multiwinner rule (see~\citet{lac-sko:b:abc-rules}), where voters express
dichotomous preferences (approval/disapproval) over the set of candidates.

Finding a winning committee for many ABC rules is \np-hard~\citep{azi-gas-gud-mac-mat-wal:c:comp-aspects-multiwinner,BrillFJL24,GodziszewskiB0F21,LeGrandMM07,pro-ros-zoh:j:proportional-representation,SkowronFL16}.
However, natural variations of these \np-hard rules exist where a winning set of candidates can be found in polynomial time. In addition, the Method of Equal Shares, which has desirable properties---such as satisfying a notion of group fairness known as Extended Justified Representation, can be computed in polynomial time~\citep{pet-sko:c:method-equal-shares}.

It is important that we can always easily compute the outcome of an election. 
Previous work on studying the complexity of winner determination has largely
focused on whether the winner problem is in \p\ or \np-hard, equating being in \p\ with being easy to compute.
However, elections can be very large (e.g., \citet{boe-bro-cev-geh-san-sch:c:abc-in-practice} note that blockchains that use the Nominated Proof-of-Stake protocol (e.g., Polkadot) use ABC rules to elect ``validators'' to ensure the integrity of the blockchain).
In such cases
it is not enough to just have a polynomial-time algorithm.
In this case parallelization may be necessary.

The parallelizability of voting rules was previously considered for single-winner voting rules~\citep{csa-lac-pic-sal:c:winner-determination-mapreduce,csa-lac-pic:c:schulze-large-scale},
but
not for the case of multiwinner rules. In fact, the recent textbook on ABC rules by \citet{lac-sko:b:abc-rules} explicitly includes as open question ``Q20'' 
 the task of determining which polynomial-time ABC rules are inherently sequential, i.e., do not have efficient parallel algorithms.
 As pointed out  by~\citet{lac-sko:b:abc-rules}, Approval Voting (AV) is clearly
 parallelizable (and the analogous result for Satisfaction Approval Voting (SAV) follows from the same argument).
We show that computing the winning committee for all other studied polynomial-time ABC rules is inherently sequential.

To show that a problem is inherently sequential we use the notion of \p-hardness, since it is generally assumed that \p-hard problems are not parallelizable \citep{gre-hov-ruz:b:limits}.
We mention that 
this approach was used previously in computational social choice
to show results about single-winner elections (specifically that winner determination for resolute versions of STV and Ranked Pairs is \p-complete~\citep{csa-lac-pic-sal:c:winner-determination-mapreduce,csa-lac-pic:c:schulze-large-scale,csa-lac-pic:t:schulze-large-scale}), 
problems related to the aggregation of CP-nets~\citep{luk-mal:j:aggregation-over-cpnets}, and for finding the essential set~\citep{bra-fis:j:minimial-covering}.
We also mention a line of research that studies the parallelizability of problems in fair allocation~\cite{gar-pso:t:fair-allocation-parallel,zhe-gar:c:distributed-algos-allocation}.

When studying
the complexity of a voting problem it is natural to consider settings where
the votes of the electorate have structure. Two important domain restrictions
are single-peaked and single-crossing preferences~\citep{bla:j:rationale-of-group-decision-making,mir:j:single-crossing}, which each model settings where the votes of the electorate are based on a single dimension (e.g., liberal vs.\ conservative). 
Both of these
models have natural counterparts for approval voting~\citep{elk-lac:c:dichotomous,fal-hem-hem-rot:j:single-peaked-preferences}.
We show that for single-peaked and for single-crossing approval votes finding a winning committee using the important Chamberlin-Courant ABC rule can be parallelized.
We note here that the only previously-known parallelizability results for voting rules that we are aware of (in addition to AV and SAV mentioned previously~\citep{lac-sko:b:abc-rules}) are the one for Schulze Voting~\citep{csa-lac-pic:c:schulze-large-scale} as well
as computing the Copeland, Smith, and Schwartz choice sets~\citep{bra-fis-har:j:choice-sets}.

Our main contributions are summarized below.

\noindent
\textbf{Inherently sequential ABC rules.}
We provide \p-hardness reductions to prove the following theorem which encompasses all studied polynomial-time ABC rules from \citet{lac-sko:b:abc-rules} except AV and SAV which are easily seen to be parallelizable \citep{lac-sko:b:abc-rules}; in fact they are computable in logarithmic space (see Section~\ref{sec:av-sav-fl} of the appendix for completeness).

\begin{theorem}
    Computing a winning committee is inherently sequential (\p-hard) for
    seq-CC,
    seq-PAV,
    rev-seq-CC, 
    rev-seq-PAV,
    seq-Phragmén,
    Greedy Monroe,
    and the Method of Equal Shares.
\end{theorem}
We show that seq-CC and the Method of Equal Shares are inherently sequential in Section~\ref{sec:p-hard-main}. The remaining rules are discussed in Section~\ref{sec:p-hard-appendix} of the appendix. In fact, we show that some of these rules are special cases of an infinite class of rules that we prove the \p-hardness of.

\smallskip

\noindent
\textbf{Parallelizability under domain restrictions.}
For the well-studied domain restrictions of single-peaked and single-crossing approval votes, we show that finding a winning committee for the important multiwinner rule Chamberlin-Courant can be parallelized (see Section~\ref{sec:domain}).

Instead of providing parallel algorithms, we show membership in the relatively unknown complexity class \optl\ from \citet{alv-jen:j:logcount} (Definition ~\ref{def:optl}). 
This allows for stronger results (since \optl\ is thought to be a strict subset of the class of parallelizable problems) and simpler proofs (since describing parallel algorithms can be an arduous task, e.g., one has to decide on a model and describe how each machine aggregates information, etc.). 

To the best of our knowledge, we are the first to use this class to show parallelizability, and we expect that this approach will be useful to further explore the parallelizability of other polynomial-time problems.

\section{Preliminaries}
\label{sec:prelim}

\subsection{Approval-Based Committee (ABC) Rules}

An approval-based committee election consists of a set of candidates $C$, a collection
of votes $V$ where each $v \in V$ is a set of approved candidates,
 and a desired
committee size $k \in \mathbb{N}$. The output of an approval-based committee rule is a set of
$k$ candidates referred to as the winning committee.

We study a large number of ABC rules, which we generally define either where they are first used or in the appendix. We present the formal definition of Chamberlin-Courant~\citep{cha-cou:j:proportional-rep} below since it will be relevant for results in Section~\ref{sec:cc} and Section~\ref{sec:domain}, and it allows us to provide an illustrative example below.

\begin{definition}[Chamberlin-Courant (CC)]
    Given candidates $C$, collection of votes $V$, and committee size $k$, 
    output a committee $W$ of size $k$ such that it
    maximizes the number of voters that have at least one approved candidate in the committee.
\end{definition}
\begin{example}
Suppose we have candidates $C = \{a,b,c\}$, committee size $k = 2$, and the following collection of votes:
\[
V = \left[ \begin{array}{cccccc}
\{a\}, & \{a,c\}, & \{a,c\}, & \{b,c\}, & \{b,c\}, & \{b\}
\end{array} \right]
\]
The winning committee computed by CC is $\{a,b\}$, since the maximum number of voters (in this case, all voters) approve of a candidate in $\{a,b\}$.
\end{example}

\subsection{Computational Complexity}

Most of our results concern showing problems to be inherently
sequential. 
It is widely believed that problems that are \p-hard
are inherently sequential~\citep{gre-hov-ruz:b:limits}. %
Note that when showing a problem is hard for a given complexity class we must be careful to use a
reduction with less computational power than the class we are showing hardness for. %
As is standard, we use logspace reductions for our \p-hardness results, meaning that the reduction is computable in logarithmic space.\footnote{Though logarithmic space may seem quite restrictive, standard \np-hardness reductions are computable in logarithmic space (see, e.g.,~\citet{gar-joh:b:int}).}

Most standard complexity classes such as \p\ and \np\ are concerned with decision problems, whose algorithms accept or reject inputs.
However, we want to output the winning committee for ABC rules, i.e., we are looking at function problems, and so we need to also look at function complexity classes.

The complexity class \nc\ (and its function counterpart \fnc) corresponds to the class of problems that have (efficient) parallel algorithms, namely, algorithms that run in polylogarithmic time 
using a polynomial number of processors (see~\citet{pap:b:complexity}).

It is easy to see that \fnc\ is a subset of \fp, the function counterpart of \p.
Our results that show problems inherently sequential are based on the generally accepted assumption that \fnc $\neq$ \fp, and we show
inclusion in \fnc\ for parallelizability. 
In fact, we show inclusion in \optl, a subset of \fnc,
which can be easier than directly presenting and analyzing parallel algorithms. To the best of our knowledge, we are the first to use \optl\ to show the parallelizability of a problem.

\section{Inherently Sequential ABC Rules}
\label{sec:p-hard-main}

This section establishes the \p-hardness of ABC rules. We highlight Chamberlin-Courant (CC) and the Method of Equal Shares (MES), while the \p-hardness results for other ABC rules are deferred to the appendix. Before we present our proofs, we discuss an important technical consideration. 

Recall that finding the winning committee\footnote{Note that the
polynomial-time ABC rules from \citet{lac-sko:b:abc-rules} explicitly include lexicographic tie-breaking, and so we can talk about \emph{the} winning committee. However, we will also show that our hardness results go through when tie-breaking is never invoked, which shows that the hardness is not caused by tie-breaking.
} is a function problem.
However, it gives more insight 
to prove the hardness of a related decision problem. 
Thus, we consider the decision problem that asks about the containment of a single candidate.
First, we show that this formulation is correct by showing the equivalence of the decision and function problem with respect to parallelizability.

\begin{observation}
    Suppose we are given a set of candidates $C$, collection of votes $V$, committee size $k$, and a candidate $c \in C$.
    For any ABC rule, finding the winning committee $W$ is parallelizable if and only if deciding whether $c$ belongs to $W$ is parallelizable.
\end{observation}
\begin{proof}
    Suppose there is a parallel algorithm $X$ 
    that outputs the winning committee $W$. Then, one can run $X$ to find $W$ and simply check whether $c \in W$ to accept/reject for the decision problem.
    Now, suppose there is a parallel algorithm $X$ that decides whether $c$ belongs to $W$. In this case, one can run $X$ for each candidate in parallel and output the $k$ candidates for which $X$ accepts.
\end{proof}

\subsection{Chamberlin-Courant}
\label{sec:cc}

Recall that the Chamberlin-Courant rule picks $k$ candidates that hit the most number of votes; a vote is ``hit'' if at least one of the candidates it approves is selected in the winning committee. %

While computing the winning committee for CC is \np-hard~\citep{pro-ros-zoh:j:proportional-representation} in the general case, 
simple greedy approaches have been introduced to mimic CC in polynomial time. One such approach is the sequential Chamberlin-Courant rule, denoted seq-CC, where we start with an empty committee and iteratively add the candidate that increases the number of currently hit votes the most (see~\citet{lac-sko:b:abc-rules}).

\begin{definition}[seq-CC]
    Suppose we are given candidates $C$, collection of votes $V$, and committee size $k$.
    seq-CC sequentially builds a committee of size $k$ by repeatedly picking the candidate that increases the number of hit votes the most.
    Ties are broken lexicographically.
\end{definition}

We will show that computing the winning committee using seq-CC is \p-hard. Specifically, we show that deciding whether a given candidate is picked for the seq-CC winning committee is \p-hard.
To do this, we give a logspace reduction from a closely related \p-complete problem, OVR, defined below.
We note that as with \np-hardness reductions, half the battle is finding the right problem to reduce from, and we were able to find a problem almost equivalent to seq-CC.
We hope that the 
straightforward nature of our reduction will
ease the reader into \p-hardness reductions.

\begin{problem}[Ordered Vertices Remaining (OVR)~\citep{gre:j:ovr}]
We are given an undirected graph $G$, a vertex $v \in V(G)$, and an integer $k$.
Suppose we repeatedly delete the highest-degree vertex from $G$ until no vertices remain---ties are broken lexicographically.
Decide whether there are at least $k$ vertices remaining after deleting $v$.
\end{problem}

\begin{theorem}
\label{thm:seqcc}
Computing the winning committee for seq-CC is inherently sequential.
\end{theorem}
\begin{proof}
Suppose we are given an instance, $(G,v,k)$, of OVR. We will construct an election with candidates $C$, collection of votes $V$, and committee size $k'$, and specify a candidate $c$ such that the winning committee selected by seq-CC includes $c$ if and only if $k$ vertices remain after $v$ is deleted. 
We provide an example of our reduction in Figure~\ref{fig:ovr-seqcc}.

Let $C$ have a candidate for each vertex of $G$,
and let
$V$ have a vote for each edge such that edge $\{a,b\} \in E(G)$ corresponds to a vote for candidates $a$ and $b$.
Finally, we set $c = v$ and $k' = \|V(G)\| - k$.

Recall that during each round of seq-CC, we pick the candidate that hits the most new votes. Similarly in OVR, we delete the vertex with highest degree. In our reduction, the vertex of the highest degree corresponds exactly to the candidate that hits the most unhit votes.
Thus, if $v$ is among the first $\|V(G)\| - k$ candidates selected with seq-CC, then there must be at least $k$ vertices remaining in $G$ after deleting $v$ in OVR. 

The reduction is clearly in logspace since we are essentially just copying the input and computing $k'$. 
\end{proof}

\begin{figure}[t]
    \centering
    \includegraphics[width=0.4\linewidth]{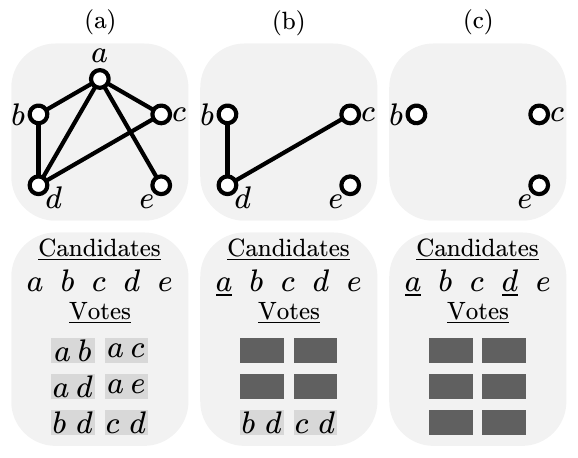}
    \caption{
    Column (a) shows the graph input to OVR at the top and the corresponding seq-CC election at the bottom. Columns (b) and (c) show the state of OVR and seq-CC after one and two rounds, respectively. 
    Candidates are underlined when they are put in the winning committee and votes are darkened when they are hit.
    Observe how the removal of the highest-degree vertex in OVR corresponds to its inclusion in the winning committee in seq-CC.}
    \label{fig:ovr-seqcc}
\end{figure}

We will see in Section~\ref{sec:thiele} of the appendix that a reduction from a different problem provides a more general result; there we analyze the \p-hardness of a class of rules, which includes seq-CC, known as sequential Thiele methods. However, we decided to include the reduction from OVR to provide a simple, illustrative \p-hardness proof in the main text.

As mentioned previously, the
polynomial-time ABC rules from \citet{lac-sko:b:abc-rules} explicitly include lexicographic tie-breaking, and so Theorem~\ref{thm:seqcc} proves that seq-CC in the definition of \citep{lac-sko:b:abc-rules} is inherently sequential (P-hard). However, one might worry that the tie-breaking is causing the hardness. We now show that that is not the case.

\begin{theorem}
    Computing the winning committee for seq-CC is inherently sequential even when tie-breaking is never invoked.
\end{theorem}
\begin{proof}
    Suppose we are given an instance of seq-CC with candidates $C = \{c_1, c_2, \ldots, c_m\}$ in lexicographic order, a collection of votes $V$, and committee size $k$.
    We construct a new collection of votes, $V'$, for which the same candidates as for $V$ will be selected, but for which tie-breaking is never invoked.
    Let $V'$ be a collection of votes with $m$ copies of each vote in $V$. In addition, we add padding votes: for candidate $c_i$, we add $m-i$ votes that approve of only $c_i$. 
    Note that the padding votes matter only in case of ties during the execution of seq-CC on the original election and so the elections select the same committee but the execution of seq-CC on the modified election never has ties.
\end{proof}

The rest of our hardness proofs reduce from a problem from which a straightforward reduction also introduces many ties. However, to avoid concerns about whether the hardness is linked to tie-breaking, we also provide proofs that ensure that tie-breaking is never invoked.

\subsection{Method of Equal Shares}

The Method of Equal Shares (MES) is a recently introduced rule~\citep{pet-sko:c:method-equal-shares} that allots voters budgets and allows them to elect candidates using shares of their budget.
A generalized version of this rule~\citep{pet-pie-sko:c:mes-partcipatory-budgeting}, has also been used in the real world for participatory budgeting (a setting that generalizes multiwinner elections to voting over issues to be funded) (see~\citet{equal-shares:m:equal-shares-elections}).

Since the definition of MES is quite involved, we first provide an informal description, followed by an example, and finally, the formal definition.

In an MES election, voters are given a budget that they can use to elect candidates they approve of, provided they can afford it. 
The election proceeds in rounds, wherein a candidate is added to the winning committee by the end of each round. Adding a candidate $c$ to the winning committee incurs a cost of 1, which is split among all voters who approve of $c$. The candidate chosen is the one that minimizes the maximum cost incurred by each voter approving them. Once a candidate is elected, each voter's budget is adjusted according to the cost they paid to elect said candidate.
The process continues until $k$ candidates have been elected, or until voters no longer possess the budget to elect a new candidate.

\begin{example}
Suppose we have candidates $\{a,b,c,d\}$
and $n = 4$ voters, labeled $1$ to $4$, who want to elect a committee of size $k = 3$. Each voter starts with an initial budget of $k/n = 3/4$. The votes and budgets of each voter are:
\[
\begin{array}{ccc|ccc}
1: & \{a,c\} & \nicefrac{3}{4} & 2: & \{a,c\} & \nicefrac{3}{4} \\ 
3: & \{a,b\} & \nicefrac{3}{4} & 4: & \{b,d\} & \nicefrac{3}{4} 
\end{array}
\]
Candidates require a cost of 1 to be elected. 
Candidate $a$ can be elected if all of its voters pay $1/3$.
Candidate $b$ can be elected if all of its voters pay $1/2$.
Candidate $c$ can be elected if all of its voters pay $1/2$.
Candidate $d$ can be elected if voter 4 pays 1, but this is not possible since the voter does not have the budget.
MES looks to minimize the maximum cost incurred by each voter; so in this case, it will pick candidate $a$ for the winning committee, since said cost is $1/3$. The cost incurred by each voter is subtracted and we are left with the following budgets:
\[
\begin{array}{ccc|ccc}
1: & \{a,c\} & \nicefrac{5}{12} & 2: & \{a,c\} & \nicefrac{5}{12} \\ 
3: & \{a,b\} & \nicefrac{5}{12} & 4: & \{b,d\} & \nicefrac{3}{4} 
\end{array}
\]
Candidate $c$ can no longer be elected since its voters, 1 and 2, have a total budget of $10/12 < 1$.
Candidate $b$'s voters have a total budget of $7/6$, so $b$ can be elected. The maximum budget paid by each of $b$'s voters is minimized if voter 4 pays $7/12$ and voter 3 pays $5/12$.
Thus, MES elects $b$, and the updated budgets are:
\[
\begin{array}{ccc|ccc}
1: & \{a,c\} & \nicefrac{5}{12} & 2: & \{a,c\} & \nicefrac{5}{12} \\ 
3: & \{a,b\} & 0 & 4: & \{b,d\} & \nicefrac{1}{6} 
\end{array}
\]
MES terminates here and returns the committee $\{a,b\}$ as the cost to elect candidates $c$ or $d$ cannot be met by its voters. Note that MES returns fewer candidates than the desired committee size. This is discussed further below.
\end{example}

We now formally define the rule, heavily based on the definition from \citet{lac-sko:b:abc-rules}.

\begin{definition}[Method of Equal Shares (MES)]
\label{def:mes}
    Suppose we are given candidates $C$, collection of votes $V$, and committee size $k$. We will assume that there are $n$ votes and each voter is identified with an integer in $[1,n]$.
    
    Let $x_r(i)$ be the budget of voter $i$ at the end of round $r$; $x_0(i) = k/n$.
    Suppose we are electing a candidate in round $r+1$.
    Let $C_r$ be the candidates that can be elected; candidates that have not already been selected for the winning committee and whose voters can afford to elect them, i.e., $\sum_{i \in N(c)} x_r(i) \geq 1$, where $N(c)$ is the set of voters that approve of $c$.
    We select a new candidate like so.
    For each candidate $c \in C_r$, we compute
    $\rho_c$: the minimum value 
    such that each voter approving $c$ pays at most $\rho_c$ and all voters pay a total of 1, i.e., we compute the minimum value $\rho_c$ satisfying:
    \[ \sum_{i \in N(c)} \min(x_r(i), \rho_c) = 1 \]
    The candidate with minimum $\rho_c$ is selected for the winning committee---ties are broken lexicographically.
    The budgets for each voter are updated as follows:
    \[
    x_{r+1}(i) = \begin{cases}
        x_{r}(i) - \rho_c & \text{if $i$ approves $c$ and $x_{r}(i) \geq \rho_c$} \\
        0 & \text{if $i$ approves $c$ and $x_{r}(i) < \rho_c$} \\
        x_r(i) & \text{otherwise}
    \end{cases}
    \]
\end{definition}

Note that voters may run out of the budget necessary to elect another candidate before $k$ rounds. In this case, the method terminates early and returns the winning committee constructed so far.
However, technically this does not meet our definition of an ABC rule, since we must return a committee of size $k$.
To rectify this, 
another ABC rule is called to select the remaining candidates---this is referred to as the second phase of MES.
\citet{lac-sko:b:abc-rules} suggest using seq-Phragmén (see Sections~\ref{sec:phrag} and~\ref{sec:mesphrag} of the appendix) since it preserves some desirable properties.

We will show that computing the committee found by (the first phase of) MES is \p-hard.
This shows that the hardness of MES is not caused by the hardness of the rule called afterward (such as seq-Phragmén).
As in the previous section, we focus on the decision problem where we ask about the inclusion of a specified candidate. We reduce from a problem known as LFMIS:

\begin{problem}[Lexicographically First Maximal Independent Set (LFMIS)~\citep{coo:j:tax}]
Suppose we are given an undirected graph $G$ and a vertex $v$.
The LFMIS of $G$ is a maximal independent set built by repeatedly picking the lexicographically first vertex not adjacent to an already picked vertex until no more vertices can be added.
Decide if $v$ belongs to the LFMIS of $G$.
\end{problem}

For reasons that will be made clear in the upcoming proof, working with regular graphs (graphs where each vertex has the same degree) allows for a much easier hardness reduction.
\cite{m:j:lfms} showed that LFMIS is \p-complete even when the input graph is subcubic, i.e., each vertex has degree at most three. In Section~\ref{sec:lfmis} of the appendix, we show how this result can be easily extended to show that LFMIS is \p-complete even when the input graph is 3-regular. Thus, for our proof below, we will assume that is the case.

\begin{figure}[t]
    \centering
    \includegraphics[width=0.5\linewidth]{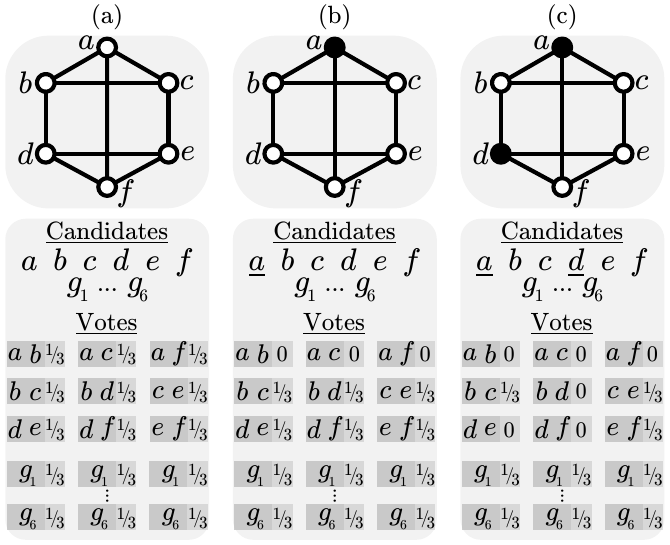}
    \caption{Column (a) shows the graph input to LFMIS at the top and the corresponding MES election at the bottom with committee size $k = 9$.
    Note the extra candidates $g_1$ to $g_6$ whose three votes each give us an initial budget of $9/27 = 1/3$. 
    The budgets of each voter are shown next to their vote.
    Observe that each candidate has exactly three voters that can pay $1/3$ for them, implying that the minimum budget $\rho_c = 1/3$ is the same for every candidate. 
    Column (b) shows the next round of both problems, where $a$ is the first vertex to be added to the LFMIS and, 
    due to lexicographic tie-breaking, candidate $a$ is the first candidate to be elected.
    Note that selected vertices are filled and selected candidates are underlined.
    The next round is depicted in column (c), where $d$ is chosen for the LFMIS since it is the first vertex not adjacent to $a$, and $d$ is elected in MES since they're the first candidate whose voters can afford to pay for them. 
    Just as vertices that share an edge with $a$ in the input graph can not be selected for the LFMIS, 
    note how the candidates that share votes with $a$ can not be elected by MES since their voters can no longer afford them.
    In the remaining rounds of the algorithm, the extra candidates $g_1$ to $g_6$ are elected by MES before it terminates. }
    \label{fig:lfmis-mes}
\end{figure}

\begin{theorem}
\label{thm:mes-main}
Computing the winning committee for (the first phase of) MES is inherently sequential.
\end{theorem}
\begin{proof}
    Suppose we are given an instance $(G, v)$ of LFMIS, where $G$ is a 3-regular graph with $m$ vertices labeled $1$ to $m$. We will construct an election with candidates $C$, collection of votes $V$, and committee size $k$, and specify a candidate $c \in C$ such that the winning committee selected by MES includes $c$ if and only if $v$ is included in the LFMIS of $G$.

    Similar to the proof of Theorem~\ref{thm:seqcc},
    let $C$ have a candidate for each vertex of $G$, and 
    $V$ have a vote for each edge such that edge $\{a,b\} \in E(G)$ corresponds to a vote for candidates $a$ and $b$.
    We also add $m$ extra candidates to $C$, labeled $m+1$ to $2m$, with three singleton votes in $V$ for each. 
    Observe that the lexicographic order of the input graph is preserved in our election.
    Note that we have a total of $n = 3m/2 + 3m$ votes.
    Finally, we set $c = v$ and $k = m/2 + m$.

    We now show that the candidates corresponding to $V(G)$ picked by MES correspond exactly to the vertices in the LFMIS of $G$.
    This is illustrated with an example in Figure~\ref{fig:lfmis-mes}.
    
    Observe that the initial budget is
    \[x_0(i) = \frac{k}{n} = \frac{m/2 + m}{3m/2 + 3m} = \frac{1}{3}\]
    for each voter $i$.
    Also, recall that each candidate has exactly three votes. At the start of round 1, any candidate $c \in C$ can be elected, since its three supporters can each pay $1/3$ to elect it. $1/3$ is also the minimum amount, $p_c$, that can be paid by each voter to achieve this.
    Due to lexicographic tie-breaking, we will select the candidate labeled 1.
    When updating our budget for the next round, observe that any voter $i$ approving candidate 1 now has a remaining budget of 0; $x_1(i) = x_0(i) - \rho_c = 1/3 - 1/3 = 0$.

    Due to this decreased budget, in the next round, any candidate that shared a vote with candidate 1 cannot be elected since 
    its voters can no longer afford to pay for them. 
    This corresponds exactly to LFMIS, wherein any vertex adjacent to an already selected vertex cannot be added to the independent set.
    Thus, in subsequent rounds, a candidate who does not share a vote with a previously elected candidate must be elected.
    It is now easy to see that the candidates corresponding to vertices in $G$
    picked by MES correspond exactly to the vertices in the LFMIS of $G$.

    Note that since $k$ is large enough (more than the number of vertices of $G$), all of the vertices in the LFMIS will be picked. The extra candidates, $m+1$ to $2m$, will be picked after LFMIS candidates due to lexicographic ordering. Note that since each extra candidate also has exactly three voters, the minimum budget to elect them, $\rho_c$, is also $1/3$.
    It should now be clear that these extra candidates were added so that our initial budget is $1/3$.

    It is clear that $v \in V(G)$ is elected by MES if and only if $v$ is part of $G$'s LFMIS, and that the reduction is clearly in logspace, thus completing our proof.
    \end{proof}

We note that the above reduction may return a winning committee with fewer than $k$ candidates. Recall that in this case another ABC rule is called to elect the remaining candidates. For the sake of completeness, 
we show in Section~\ref{sec:mesphrag} of
our appendix that MES is also 
P-hard when seq-Phragmén is used, as suggested by~\citet{lac-sko:b:abc-rules}, to finish electing a total of $k$ candidates.
Particularly, we reduce from LFMIS and show that the specified vertex $v$ is picked by MES if and only if $v$ belongs to the LFMIS.

As in the case of seq-CC we can show that (the first phase of) MES is \p-hard for the case where tie-breaking is never invoked. Informally, we pad our reduction from LFMIS with enough votes so that none of the candidates are tied initially, and through the execution of MES at no point will candidates become tied.
See Section~\ref{sec:mes-noties} of the appendix for details.

\begin{theorem}\label{mes:noties}
Computing the winning committee for (the first phase of) MES is inherently sequential even when tie-breaking is never invoked.
\end{theorem}

\section{Domain Restrictions}
\label{sec:domain}

We now turn our attention to domain restrictions. Such restrictions often result
in axiomatic and computational advantages (see, e.g.,~\citet{elk-lac-pet:b:structured}).
The most commonly studied domain restrictions are single-peaked preferences~\citep{bla:j:rationale-of-group-decision-making} and single-crossing preferences~\citep{mir:j:single-crossing}. 
Each of these restrictions capture
natural settings where the electorates are focused on a single issue.
In single-peakedness, this results in an ordering of the candidates from one extreme to the other, while in a single-crossing electorate the voters are ordered with respect to this issue and voters on either end represent the extremes. %
These restrictions were each first defined for preference order ballots, but have natural definitions for approval votes.
\citet{fal-hem-hem-rot:j:single-peaked-preferences}
introduced
the model of single-peakedness for approval ballots, and
\citet{elk-lac:c:dichotomous} introduced the
model of single-crossingness for approval ballots.

\begin{figure}[t]
    \centering
    \includegraphics[width=0.5\linewidth]{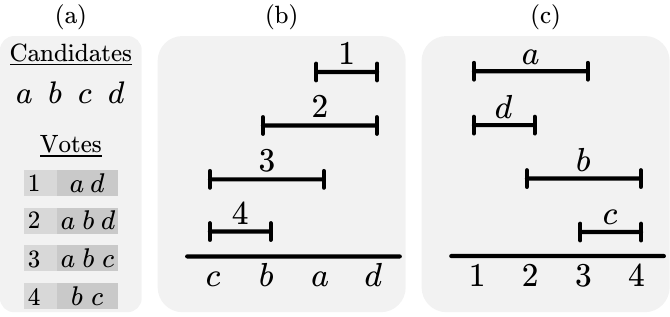}
    \caption{Column (a) shows an election. Column (b) shows how the votes are single-peaked with respect to the axis $c\ b\ a\ d$. Column (c) shows how the votes are single-crossing with respect to the axis $1\ 2\ 3\ 4$. The intervals (candidates) in Column (c) are ordered in nondecreasing order with respect to starting point, as will be done in the construction of Theorem \ref{sccc}. } 
    \label{fig:spsc-example}
\end{figure}

A collection of (approval) votes is single-peaked if there exists a linear order of the candidates (a single-peaked axis) such that each vote forms a contiguous interval on the axis~\citep{fal-hem-hem-rot:j:single-peaked-preferences} (this is also known as candidate interval (CI)~\citep{elk-lac:c:dichotomous}).
A collection of (approval) votes is single-crossing if there exists a linear order of the voters (a single-crossing axis) such that for each candidate the set of voters that approve that candidate forms a contiguous interval on the axis (this is also known as voter interval (VI))~\citep{elk-lac:c:dichotomous}. See Figure~\ref{fig:spsc-example} for an example of an
election whose votes are both single-peaked and single-crossing.

As pointed out by \citet{fal-hem-hem-rot:j:single-peaked-preferences} (for single-peaked) and \citet{elk-lac:c:dichotomous} (for single-crossing), the problem of determining whether a collection of votes (in our setting of approval ballots) is single-peaked or single-crossing, and computing an axis witnessing this is in essence the 
problem of determining whether a 0-1 matrix has the consecutive ones property and computing a permutation that witnesses that. These problems have long been known to be computable in polynomial time~\citep{boo-lue:j:pqtrees,ful-gro:j:incidence-matrices}.

Algorithms for these domain restrictions typically start with computing an axis.
Since we are interested in complexity classes below \p, we need to look at the complexity of computing an axis in more detail.
Careful inspection of the literature shows that computing whether a 0-1 matrix has the consecutive ones property and computing a permutation that witnesses that is in fact computable in logarithmic space~\citep{DBLP:journals/siamcomp/KoblerKLV11}.
This immediately gives us the novel observation that computing whether a collection of votes (in our setting of approval ballots) is single-peaked or single-crossing and computing an axis witnessing this is computable in logarithmic space. We note that this claim also holds in the common setting where the votes are preference orders. This follows from observing that the reductions to the consecutive ones problem~\citep{bar-tri:j:stable-matching-from-psychological-model,bre-che-woe:j:char-single-crossing} are easily seen to be computable in logarithmic space.

\begin{observation}
Computing whether a collection of votes is single-peaked or single-crossing and computing an axis witnessing this is computable in logarithmic space. This holds in our setting of approval ballots as well as in the setting of preference order ballots.
\end{observation}

As mentioned previously, domain restrictions can lower complexity. This is also the case in our context of approval ballots. For example, finding a winning committee for Chamberlain-Courant is \np-hard~\citep[Theorem 1]{pro-ros-zoh:j:proportional-representation}, 
but for single-peaked votes and for single-crossing votes, 
this problem is computable in polynomial time~\citep{bet-sli-uhl:j:proportional-representation,elk-lac:c:dichotomous}.

We use SP-CC (resp., SC-CC) to refer to the problems of computing a CC committee when the votes are single-peaked (resp., single-crossing).
In contrast to the results of the previous section, we will show that these problems are even parallelizable.

\begin{theorem}
\label{t:spsccc}
SP-CC and SC-CC are parallelizable.
\end{theorem}

We will show that SP-CC and SC-CC are in some sense computable in nondeterministic logarithmic space, in a way that will imply that these problems are parallelizable.
When looking at deterministic complexity classes, it is a little sloppy
but not dangerous to talk about function problems (such as computing the winning coalition) as being in a complexity class like P (really a class of decision problems). However, we need to be very careful when talking about nondeterministic function classes, since such functions are inherently multi-valued, while we are interested in single-valued functions. 
There are multiple ways to get from multi-valued to single-valued functions. We will use the simplest way for our purposes: the function class \optl~\citep{alv-jen:j:logcount}.

\begin{figure}[t]
\centering
\includegraphics[width=0.5\linewidth]{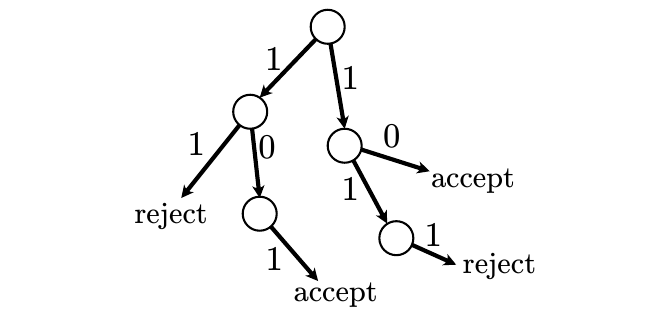}
    \caption{This figure depicts the computation tree of a transducer on some input. The output values of the accepting paths of this transducer are 101 and 10. The value of the \optl\ function defined by this transducer given this input is the maximum of these values, i.e., 101.}
    \label{fig:optl-example}
\end{figure}

\begin{definition}[\optl]\label{def:optl}
\optl\ is the class of functions computable by taking the maximum of the output values over all accepting paths of an \nl\ transducer, i.e., a nondeterministic Turing machine with a read-only input tape, a logspace work tape, and a write-only output tape.
\end{definition}

We provide a visualization in Figure~\ref{fig:optl-example}.
It is known that every \optl\ function is parallelizable~\cite[Theorem 4.2]{alv-jen:j:logcount}. 
We will show that SP-CC and SC-CC are in \optl\ (Theorems~\ref{spcc} and \ref{sccc}).
Since \optl\ is a subset of the class of parallelizable problems, this immediately implies
Theorem~\ref{t:spsccc}.
As an aside, a benefit of this approach is that we can show parallelizability without needing to directly
describe a parallel algorithm, which can be quite involved (e.g., we would first need to decide on a model, describe how each machine aggregates information, we would need to show polylogarithmic running time, etc).

Since we can compute an SC or SP axis in logarithmic space, we will see that our problems are in essence interval problems, with SP-CC corresponding to computing a set of $k$ points (candidates) that hit the most intervals (voters) and SC-CC corresponding to computing a set of $k$ intervals (with each interval being the set of voters that approve a particular candidate) whose union (the set of voters that approve at least one of those candidates) is as large as possible.

\begin{theorem}
\label{spcc}
    SP-CC is in \optl.
\end{theorem}

\begin{proof}
Assume that there are $m$ candidates and $n$ voters. 
Recall that a collection of votes is single-peaked if there exists a linear order of the candidates (a single-peaked axis) such that each vote forms a contiguous interval on the axis. Let $L$ be the single-peaked 
axis given by the logarithmic space algorithm that computes an SP axis.

For $1 \leq i \leq n$, let $[s_i,t_i]$ be the interval corresponding to the $i$th voter, i.e., $s_i$ and $t_i$ are in $\{1, \ldots, m\}$ and the $i$th voter approves the candidates from the $s_i$th candidate to the $t_i$th candidate on our SP axis $L$. 
Note that we can't store all these values, but we can recompute each value in logarithmic space whenever it is needed.

We now need to compute a set of $k$ candidates that hit (have nonempty intersection with) the most voters (intervals). We only keep track of the last candidate ($lastc$) and the next candidate chosen ($nextc$),
the number of candidates selected so far ($numc$), the number of voters hit so far ($numv$), and a counter to iterate through the voters ($i$). We also need to make sure that the maximum output value corresponds to a set of $k$ candidates that hits the most voters ($optv$).
Our algorithm is presented in Algorithm~\ref{alg:spcc}.

\begin{algorithm}[tb]
    \caption{SP-CC}
    \label{alg:spcc}
    \begin{algorithmic}[1] %
    \STATE guess $optv \leq n$
    \STATE \textbf{append} $1^{optv \cdot m(m+1)}0$ to output
    \STATE $lastc = 0$; $numv = 0$
    \FOR{$numc = 1$ to $k$}
        \STATE guess $nextc$ such that $lastc < nextc \leq m$
        \STATE \textbf{reject} if $lastc \geq m$
        \STATE $numv$ += $\| \{ [s_i,t_i] \ | \ s_i > lastc \text{ and } nextc \in [s_i,t_i] \} \|$ 
        
        \label{alg:spcc:numv}
        
        \STATE $lastc = nextc$
        \STATE \textbf{append} $1^{lastc}0$ \ to output \label{alg:spcc:output}
    \ENDFOR
    \STATE \textbf{accept} if $numv = optv$, \textbf{reject} otherwise 
    \end{algorithmic}
\end{algorithm}

Note that every interval that is hit is counted exactly once, namely the first time we choose a candidate that is in that interval. Line~\ref{alg:spcc:numv}
ensures that later iterations do not count this interval again. Also note that the total length of the output produced in Line~\ref{alg:spcc:output}
is clearly upper bounded by $m(m+1)$, and so it follows that larger values of $optv$ will always give larger output values. Thus, the maximum output on an accepting path will correspond to a solution of the SP-CC problem.
\end{proof}

\begin{theorem}
\label{sccc}
    SC-CC is in \optl.
\end{theorem}

\begin{proof}
Assume that there are $m$ candidates and $n$ voters. 
Recall that a collection of votes is single-crossing if there exists a linear order of the voters (a single-crossing axis) such that for each candidate the set of voters that approve that candidate forms a contiguous interval on the axis. Let $L$ be the single-crossing 
axis given by the logarithmic space algorithm that computes an SC axis.

For $1 \leq i \leq m$, let $[s_i,t_i]$ be the interval corresponding to the $i$th candidate in nondecreasing order of $s_i$, i.e., $s_i$ and $t_i$ are in $\{1, \ldots, n\}$ and the $i$th candidate (in an order where the candidates are ordered in nondecreasing order of $s_i$) is approved by the voters from the $s_i$th voter to the $t_i$th voter on our SC axis $L$.
Note that we can't store all these values, but we can recompute each value in logarithmic space whenever it is needed.

We need to compute a set of $k$ intervals (candidates)
that cover the most voters, i.e., whose union is the largest. Note that we can't simply guess $k$ intervals and compute the size of the union, since we do not have the space to store $k$ intervals. But we can limit what we need to store when going through the intervals in nondecreasing order of start ($s_i$). We only keep track of the last interval ($lastc$) and the next interval chosen ($nextc$),
the number of intervals (candidates) selected so far ($numc$), the number of voters covered so far ($numv$), and the last voter covered ($lastv$). We also need to make sure that the largest output corresponds to a set of $k$ intervals that covers the most voters ($optv$).
Our algorithm is presented in Algorithm~\ref{alg:sccc}.

\begin{algorithm}[tb]
    \caption{SC-CC}
    \label{alg:sccc}
    \begin{algorithmic}[1] %
    \STATE guess $optv \leq n$
    \STATE \textbf{append} $1^{optv \cdot m(m+1)}0$ to output
    \STATE $lastv = 0$; $lastc = 0$; $numv = 0$
    \FOR{$numc = 1$ to $k$}
        \STATE guess $nextc$ such that $lastc < nextc \leq m$
        \STATE \textbf{reject} if $lastc \geq m$
        \STATE $numv $ += $ \left\| [lastv + 1, n] \cap [s_{nextc},t_{nextc}] \right\|$ 
        \STATE \textbf{append} $1^{nextc}0$ to output \label{alg:sccc:output}
        \STATE $lastc = nextc$
        \IF{$t_{lastc} > lastv$}
            \STATE $lastv = t_{lastc}$
        \ENDIF
    \ENDFOR
    \STATE \textbf{accept} if $numv = optv$, \textbf{reject} otherwise 
    \end{algorithmic} 
\end{algorithm}
Note that the total length of the output produced in Line~\ref{alg:sccc:output} is clearly upper bounded by $m(m+1)$, and so it follows that larger values of $optv$ will always give larger output values. Thus, the maximum output on an accepting path will correspond to a solution of the SC-CC problem.
\end{proof}

\section{Conclusion and Future Work}

We systematically studied the parallelizability of
finding a winning committee for all of the polynomial-time Approval-Based Committee rules studied by~\citet{lac-sko:b:abc-rules} in
their recent textbook. We find that with the exception
of two simple cases, as pointed out by~\citet{lac-sko:b:abc-rules}, all of the remaining rules are inherently sequential. We further explored the parallelizability of ABC rules by considering restricted domains and found  that for the natural
settings of single-peaked and single-crossing votes
finding a Chamberlin-Courant committee is parallelizable. We showed this result by giving algorithms that show these problems are in the complexity class OptL, the first results of this type in computational social choice.

There are clear directions for future work. It would be interesting to see which other inherently sequential ABC rules 
are parallelizable under domain restrictions.
In general, since many real-world elections can be quite large, further study of the parallelizability of election problems is warranted.

\section*{Acknowledgments}

We would like to thank Dominik Peters, Tomoyuki Yamakami, and the anonymous reviewers for their helpful comments and suggestions.
This work was supported in part by NSF grants CCF-2421977 and CCF-2421978.

\appendix

\section{Logspace Algorithms for AV and SAV}
\label{sec:av-sav-fl}

For the sake of completeness---to ensure we study all polynomial-time rules in~\citet{lac-sko:b:abc-rules}, we show in this section that Approval Voting and Satisfaction Approval Voting are in FL (the class of functions computable in logarithmic space).

\begin{definition}
    [Approval Voting (AV)]
    Given candidates $C$, collection of votes $V$, 
    and committee size $k$ 
    output the winning committee of AV, i.e., the $k$ candidates with the most votes. Formally, find the committee $W$ such that 
    \[\sum_{ v \in V } \|W \cap v\|\]
    is maximized.
    Ties are broken lexicographically.
\end{definition}

Satisfaction Approval Voting (SAV) is defined similarly, where the function being maximized is instead 
\[ \sum_{ v \in V } \frac{\|W \cap v\|}{\|v\|}\]

\begin{theorem}
(Satisfaction) Approval Voting is in \fl.
\end{theorem}
\begin{proof}
    (Satisfaction) Approval Voting simply picks the set of $k$ candidates with the most (weighted) votes; the order in which the committee is built does not matter.
    We provide an algorithm for Approval Voting in Algorithm~\ref{alg:avlogspace}. Observe that 
    there are a constant number of variables and they can be stored in logspace.
    It is easy to see that this algorithm can be modified to obtain a logspace algorithm for SAV.
\end{proof}
\begin{algorithm}[h]
    \caption{Approval Voting}
    \label{alg:avlogspace}
    \textbf{Input}: Candidates $C$, collection of voters $V$, and committee size $k$. \\
    \textbf{Output}: Committee of AV of size $k$.
    
    \begin{algorithmic}[1] %
    \FOR{$c \in C$}
        \STATE count votes for $c$ and store in $cvotes$
        \STATE let $rank = 0$
        \FOR{$d \in C$}
            \STATE count votes for $d$ and store in $dvotes$
            \IF{ $cvotes > dvotes$ \\ \hfill or ($cvotes = dvotes$ and $c < d$)}
            \STATE $rank$ += $1$
        \ENDIF
        \ENDFOR
        
        \IF{ $rank > \|C\| -k$}
            \STATE \textbf{append} $c$ to output
        \ENDIF
    \ENDFOR
    \end{algorithmic}
\end{algorithm}

\section{LFMIS is \p-complete for 3-regular Graphs}
\label{sec:lfmis}

Since many of our proofs become simpler with a 3-regular version of LFMIS, we prove its P-completeness here.

\begin{figure}[h]
    \centering
    \includegraphics[width=0.5\linewidth]{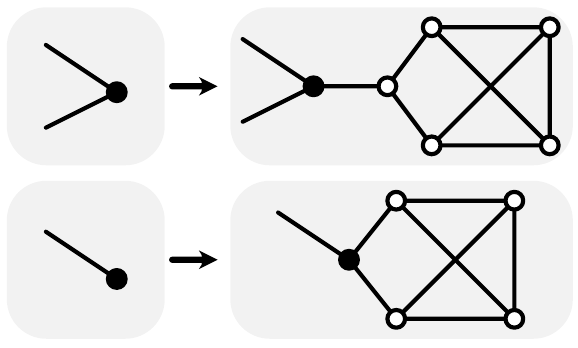}
    \caption{In this figure, we show how to append subgraphs to degree 1 and degree 2 vertices for our proof that LFMIS is \p-complete even when the input graph is 3-regular.}
    \label{fig:lfmis3reg}
\end{figure}

\begin{lemma}
    LFMIS is \p-complete even when the input graph is 3-regular.
\end{lemma}
\begin{proof}
    \cite{m:j:lfms} showed that LFMIS is \p-complete even when each vertex in the input graph has degree at most 3.
    Suppose $(G,v)$ is the input to LFMIS. We construct a graph $H$, that has $G$ as a subgraph, where each vertex has degree exactly 3 and $v$ is in the LFMIS of $G$ if and only if $v$ is in the LFMIS of $H$.

    We construct $H$ by taking a copy of $G$ and appending a subgraph to each degree 1 and degree 2 vertex. We show these subgraphs in Figure~\ref{fig:lfmis3reg}. 
    Observe that all vertices in $H$ now have degree 3.
    The vertices of these subgraphs will be labeled so that they appear lexicographically after the vertices in $G$. Thus, the existence of these new vertices will not influence whether or not $v$ will be added to the LFMIS. The reduction is clearly in logspace, so our statement holds.
\end{proof}

\section{Inherently Sequential ABC Rules}
\label{sec:p-hard-appendix}

In this section, we show the remaining proofs that all the polynomial-time rules (except AV and SAV) studied by \citet{lac-sko:b:abc-rules} are inherently sequential. 
Surprisingly for many of the rules we use the same reduction (though the proofs of correctness are, of course, different).
Definitions and examples for each rule can be found in the book by~\citet{lac-sko:b:abc-rules}. For completeness, we provide our own definitions below as well.

\subsection{Sequential Thiele Methods}
\label{sec:thiele}

Recall that in Section~\ref{sec:cc}, we showed the \p-hardness of seq-CC. 
seq-CC is part of a general family of ABC rules known as sequential Thiele methods. 
For a collection of votes $V$ for candidates $C$,
Thiele methods find a committee $W \subseteq C$ of size $k$ that optimizes a function of the form 
\[\sum_{v \in V} w \left( \| W \cap v \| \right)\] 
for some nondecreasing function $w$.
For example, the function $w$ corresponding to 
seq-CC is $w(x) = \min \left(1,x\right)$.

The greedy approach to build a committee that optimizes this function is called seq-$w$-Thiele, and is defined similarly to seq-CC.

\begin{definition}[seq-$w$-Thiele]
    Let $w: \mathbb{Z}_{\geq 0} \xrightarrow{} \mathbb{R}$ be a fixed nondecreasing function with $w(0) = 0$.
    Suppose we are given candidates $C$, collection of votes $V$, and committee size $k$.
    seq-$w$-Thiele sequentially builds a committee, $W$, of size $k$ by repeatedly 
    adding to $W$ 
    the candidate whose inclusion increases the value
    \[\sum_{v \in V} w \left( \| W \cap v \| \right)\] 
    the most. 
    Ties are broken lexicographically.
\end{definition}

Another popular Thiele method is Proportional Approval Voting (PAV). In PAV, a voter's contribution to the score of the committee is based on how many approved candidates it voted for; if it voted for one candidate that was approved, it contributes a score of $1$, if it voted for 2, it contributes a score of $1 + 1/2$, and so on. This scoring is captured by the harmonic function:
$h(x) = \sum_{k = 1}^x 1/x$.
In \citet{lac-sko:b:abc-rules}, seq-$h$-Thiele is referred to as seq-PAV.
Below, we prove that computing seq-$w$-Thiele is inherently sequential for any rule where we have $w(0) = 0$, $w(1) = 1$, and $w(2) = 1+\epsilon$, where $\epsilon \in [0,1)$ is a rational number. Note that seq-CC and seq-PAV both fit this class of Thiele methods.

\begin{theorem}
    Computing the winning committee for seq-$w$-Thiele is inherently sequential for any nondecreasing function $w$ with $w(0) = 0$, $w(1) = 1$, and $w(2) = 1 + \epsilon$ for any rational $\epsilon \in [0,1)$, even when tie-breaking is never invoked.
\end{theorem}
\begin{proof}
    Suppose we are given an instance $(G, v)$ of LFMIS, where $G$ is a 3-regular graph with $m$ vertices labeled $1$ to $m$. We will construct an election with candidates $C$, collection of votes $V$, and committee size $k$, and specify a candidate $c \in C$ such that the winning committee selected includes $c$ if and only if $v$ is included in the LFMIS of $G$.
    
    Let $C$ have a candidate for each vertex of $G$.
    Let $p := \lceil \frac{4}{1-\epsilon} \rceil$.
    For each $\{a,b\} \in E(G)$, we add $pm$ votes   
    $\{a,b\}$ to $V$.
    We also add $(m-a)$ singleton votes for each $a \in C$.
    Note that each candidate $a \in \{1, 2, \ldots, m\}$ has $3pm + (m-a)$ votes. 
    Finally, add $m$ extra candidates, labeled $m+1$ to $2m$, that have $3pm-1, 3pm-2, \ldots, 3pm-m$ votes.
    Note that these extra candidates have fewer votes than all the candidates corresponding to vertices in $G$.
    It is clear that there are no ties in the election.
    Finally, we set $c = v$ and $k = m$.
    We now show that the candidates picked from $V(G)$ correspond exactly to the vertices in the LFMIS of $G$. 

    \begin{claim}
    \label{claim:thiele-1}
        Let $a$ and $b$ be unelected candidates in round $r + 1$. If $b$ shares votes with a previously elected candidate and $a$ does not, then adding $a$ to $W$ increases the score more than adding $b$.
    \end{claim}
    \begin{proof}

        Any candidate $b$ that shares votes with a previously elected candidate gives an increase in score of
        \[ 2pm +(m-b) +  \epsilon pm \leq (2p+1+\epsilon p)m  \]
        Any candidate $a$ that shares no votes with a previously elected candidate has at least $3pm-m$ votes.

        We will show that $3pm - m$ is greater than $(2p+1+\epsilon p)m$. To do so, it is sufficient to show that 
        \begin{align*}
            3pm - m &> (2p+1+\epsilon p)m \\
            pm &> 2m + \epsilon pm \\
            (p- \epsilon p)m &> 2m \\
            p - \epsilon p &> 2 \\
        \end{align*}

    Since $ x + 1 > \lceil x \rceil \geq x  $ for any $x$, we have:
        
        \[ p - \epsilon p = \left\lceil \frac{4}{1-\epsilon} \right\rceil - \epsilon \left\lceil \frac{4}{1-\epsilon} \right\rceil \geq \frac{4}{1-\epsilon} - \epsilon \left( \frac{4}{1-\epsilon} +1 \right) = 4 - \epsilon \]

        We have $4 - \epsilon > 2 \implies 2 > \epsilon$ by definition of $\epsilon$.
    \end{proof}

    By the claim above and the fact that candidates with more votes (i.e., vertices with earlier lexicographic order) are preferred, it is clear that the candidates picked will correspond exactly to the vertices in the LFMIS of $G$.
    It is also clear that the tie-breaking rule is never invoked while building the committee.

    Once adding candidates from $V(G)$ is suboptimal, we will start picking the candidates labeled $m+1$ to $2m$ until the algorithm terminates. 
    
    Note that $k$ is large enough (at least the number of vertices in $G$) so that every candidate corresponding to a vertex in the LFMIS of $G$ will be picked. Similarly, the number of extra candidates is large enough so that we are not forced to pick any candidates corresponding to the vertices in $G$ after the LFMIS candidates have been picked.
    
    Therefore, $v$ is picked in the committee by seq-$w$-Thiele if and only if $v$ is part of the LFMIS. The reduction is clearly in logspace.
\end{proof}

Another way to build committees while trying to optimize the function described above is to start by selecting every candidate for the winning committee and remove the candidate whose removal decreases the objective function the least. For a function $w$, this method is referred to as rev-seq-$w$-Thiele:
\begin{definition}[rev-seq-$w$-Thiele]
    Let $w: \mathbb{Z}_{\geq 0} \xrightarrow{} \mathbb{R}$ be a fixed nondecreasing function with $w(0) = 0$.
    Suppose we are given candidates $C$, collection of votes $V$, and committee size $k$.
    rev-seq-$w$-Thiele builds a committee, $W$, of size $k$ by first adding every candidate to $W$ and then 
    repeatedly removing from $W$
    the candidate whose removal decreases the value
    \[\sum_{v \in V} w \left( \| W \cap v \| \right)\] 
    the least. 
    Ties are broken lexicographically.
\end{definition}

\begin{theorem}
    Computing the winning committee for rev-seq-$w$-Thiele is inherently sequential for any nondecreasing function $w$ with $w(0) = 0$, $w(1) = 1$, and $w(2) = 1 + \epsilon$ for any rational $\epsilon \in [0,1)$, even when tie-breaking is never invoked.
\end{theorem}
\begin{proof}
    We reduce from the complement of LFMIS, which instead asks if the given vertex $v$ does not belong to the LFMIS of $G$.

    Suppose we are given an instance $(G, v)$ of $\overline{\text{LFMIS}}$, where $G$ is a 3-regular graph with $m$ vertices labeled $1$ to $m$. We will construct an election with candidates $C$, collection of votes $V$, and committee size $k$, and specify a candidate $c \in C$ such that the winning committee selected by rev-seq-$w$-Thiele does not include $c$ if and only if $v$ is included in the LFMIS of $G$.

    Let $C$ have a candidate for each vertex of $G$.
    Let $p := \lceil \frac{4}{1-\epsilon} \rceil$.
    For each $\{a,b\} \in E(G)$, we add $pm$ votes
    $\{a,b\}$ to $V$.
    We also add $a$ singleton votes for each $a \in C$.
    Note that each candidate $a \in \{1, 2, \ldots, m\}$ has $3pm + a$ votes. 
    Finally, add $2m$ extra candidates, labeled $m+1$ to $3m$.

    We add $3pm$ votes $\{m+1, m+2\}$, $3pm$ votes $\{m+3, m+4\}$, $\ldots$, and $3pm$ votes $\{3m-1, 3m\}$.
    We also add $a$ singleton votes for each $a \in \{ m+1, m+2, \ldots, 3m \}$. So each extra candidate $a$ has $3pm + a$ votes.
    Finally, let $k = 3m - m = 2m$ and $c=v$.

    We now show that the candidates removed from $V(G)$ correspond exactly to the vertices in the LFMIS of $G$.

    \begin{claim}
        \label{clm:rev-seq-thiele-1}
        Let $a$ and $b$ be selected candidates in round $r+1$. If $b$ shares votes with a previously removed candidate and $a$ does not, then removing $a$ from $W$ decreases the score less than removing $b$.
    \end{claim}
    \begin{proof}
        Removing $b$ gives a decrease of 
        \[ \epsilon(2pm) + pm + b \geq \epsilon(2pm) + pm \]
        Removing $a$ gives a decrease of
        \[\epsilon(3pm) + a \leq \epsilon(3pm) + 3m \]

        We will show that 
        $\epsilon(3pm) + 3m < \epsilon(2pm) + pm $. To do so, it is sufficient to show that 
        \begin{align*}
            \epsilon(2pm) + pm &> \epsilon(3pm) + 3m \\
            pm - \epsilon pm &>  3m \\
            p - \epsilon p &>  3
        \end{align*}

    Since $ x + 1 > \lceil x \rceil \geq x  $ for any $x$, we have:
        
        \[ p - \epsilon p = \left\lceil \frac{4}{1-\epsilon} \right\rceil - \epsilon \left\lceil \frac{4}{1-\epsilon} \right\rceil \geq \frac{4}{1-\epsilon} - \epsilon \left( \frac{4}{1-\epsilon} +1 \right) = 4 - \epsilon \]

        We have $4 - \epsilon > 3 \implies 1 > \epsilon$ by definition of $\epsilon$.
    \end{proof}

    By the claim above and the fact that candidates with fewer votes (i.e., vertices with earlier lexicographic order) are preferred, it is clear that the candidates removed will correspond exactly to the vertices in the LFMIS of $G$.
    It is also clear that the tie-breaking rule is never invoked while building the committee.
    
    Once removing candidates from $V(G)$ is suboptimal, we will start removing the candidates labeled $m+1$ to $3m$ until the algorithm terminates. 
    
    Note that $k$ is set so that $m$  candidates will be removed, i.e., every candidate corresponding to a vertex in the LFMIS of $G$ will be removed. Similarly, the number of extra candidates is large enough so that we are not forced to remove any candidates corresponding to the vertices in $G$ after the LFMIS candidates have been picked.

    Therefore, $v$ is removed from the committee by rev-seq-$w$-Thiele if and only if $v$ is part of the LFMIS; $v$ is part of the final committee if and only if it is not part of the LFMIS. The reduction is clearly in logspace.
\end{proof}

\subsection{seq-Phragmén}
\label{sec:phrag}

A reduction similar to the one given for seq-$w$-Thiele can be used to prove the \p-hardness of seq-Phragmén, a rule which builds a committee while trying to balance the ``load'' each voter bears for approving a candidate.
We now formally describe the process with which seq-Phragmén picks a committee.

\begin{definition}[seq-Phragmén]
    Suppose we are given candidates $C$, collection of votes $V$, and committee size $k$. We will assume that there are $n$ votes and each voter is identified with an integer in $[1,n]$.

    Let $y_r(i)$ be the load assigned to voter $i$ in round $r$ and $W$ be the set of candidates approved so far. Initially, $W = \emptyset$ and $y_0(i) = 0$ for all $i$. In the $r$th round, we calculate for each candidate the maximum load that would arise from including it in $W$:
    
    \[
    \ell_r(c) := \frac{1 + \sum_{i \in N(c)} y_{r-1}(i) }{\|N(c)\|}
    \]
    where $N(c)$ is the set of voters that approve of $c$.    
    The candidate, $c_r$ with the minimum $\ell_r(\cdot)$ is then added to $W$, and the loads of each voter are adjusted like so:
    
    \[
    y_r(i) = \begin{cases}
        \ell_r(c_r) & \text{if } i \in N(c_r) \\
        y_{r-1}(i) & \text{if } i \not\in N(c_r) 
    \end{cases}
    \]
    
\end{definition}

The proof for seq-Phragmén is quite similar to that of seq-$w$-Thiele.

\begin{theorem}
    Computing the winning committee for seq-Phragmén is inherently sequential, even when tie-breaking is never invoked.
\end{theorem}
\begin{proof}
    Suppose we are given an instance $(G, v)$ of LFMIS, where $G$ is a 3-regular graph with $m$ vertices labeled $1$ to $m$. Without loss of generality, we will assume that $m \geq 3$. We will construct an election with candidates $C$, collection of votes $V$, and committee size $k$, and specify a candidate $c \in C$ such that the winning committee selected by seq-Phragmén includes $c$ if and only if $v$ is included in the LFMIS of $G$.

    Our construction is 
    similar to the proof for Theorem~\ref{thm:mes-main}.
    Let $C$ have a candidate for each vertex of $G$.
    For each $\{a,b\} \in E(G)$, we add $m^2$ votes $\{a,b\}$ to $V$.
    We also add $(m-a)$ votes for each $a \in C$.
    Note that each candidate $a \in \{1, 2, \ldots, m\}$ has $\| N(a) \| = 3m^2 + (m-a)$ votes. 
    Finally, add $m$ extra candidates, labeled $m+1$ to $2m$, that have $3m^2-1, 3m^2-2, \ldots, 3m^2 -m$ votes.
    It is clear that there are no ties in the election.
    Finally, we set $c = v$ and $k = m$.
    
    We now show that the candidates picked from $V(G)$ correspond exactly to the vertices in the LFMIS of $G$.
    We first prove the following claim.

    \begin{claim}
    \label{clm:seq-phrag}
        Let $a$ and $b$ be unelected candidates in round $r$. If $b$ shares votes with an elected candidate and $a$ does not, then $\ell_r(b) > \ell_r(a)$.
    \end{claim}
    \begin{proof}
        Note that we have 
        
        \[ \ell_r(a) = \frac{1 + \sum_{i \in N(a)} y_{r-1} (i)}{\|N(a)\|} = \frac{1}{\|N(a)\|} \leq \frac{1}{3m^2 -m} \]
        and 
        Suppose $b$ shares votes with a previously elected candidate, $c$. The load of all the votes that $b$ and $c$ share is at least $1/\|N(c)\|$. There are $m^2$ such votes. Thus, we have:

        \begin{align*}
            \ell_r(b) &= \frac{1 + \sum_{i \in N(b)} y_{r-1} (i)}{\|N(b)\|} \\
            &\geq \frac{1+ m^2 \left(\frac{1}{\|N(c)\|}\right)}{\|N(b)\|} \\
            &\geq \frac{1+ m^2 \left(\frac{1}{3m^2 + m}\right)}{3m^2 + m}
        \end{align*}

        It can be verified that 
        \[ \frac{1+ m^2 \left(\frac{1}{3m^2 + m}\right)}{3m^2 + m} > \frac{1}{3m^2 -m} \]
        for all $m \geq 3$.
    \end{proof}

    In the first round, the candidate with the most votes is preferred since it minimizes the maximum load. Thus, candidate $1$ is picked. 
    By Claim~\ref{clm:seq-phrag} and the fact that candidates with more votes (i.e., vertices with earlier lexicographic order) are preferred, it is clear that the candidates picked after candidate 1 will correspond exactly to the vertices in the LFMIS of $G$.
    
    Once adding candidates from $V(G)$ is suboptimal, we will start picking the candidates labeled $m+1$ to $2m$ until the algorithm terminates. 
    
    Note that $k$ is large enough (at least the number of vertices in $G$) so that every candidate corresponding to a vertex in the LFMIS of $G$ will be picked. Similarly, the number of extra candidates is large enough so that we are not forced to pick any candidates corresponding to the vertices in $G$ after the LFMIS candidates have been picked.
    
    Therefore, $v$ is picked in the committee by seq-Phragmén if and only if $v$ is part of the LFMIS. The reduction is clearly in logspace.
\end{proof}

\subsection{Greedy Monroe}

Greedy Monroe sequentially elects candidates 
by repeatedly finding representatives for groups of voters and electing the candidate with the most votes among unrepresented voters. We provide a formal definition below.

\begin{definition}[Greedy Monroe (GM)]
    Suppose we are given candidates $C$, collection of votes $V$, and committee size $k$. We will assume that there are $n$ votes and each voter is identified with an integer in $[1,n]$.

    Greedy Monroe picks candidates in $k$ rounds.
    Let $V_r$ be the set of unrepresented voters at the start of round $r$; $V_1 = \{1,2,\ldots, n\}$.
    At the start of round $r+1$, GM finds the unelected candidate $c_{r+1}$ with the most number of voters among $V_{r+1}$ (ties are broken lexicographically).
    $c_{r+1}$ is added to the winning committee. Moreover, it is chosen as the representative for a subset of its voters. Let $H$ be the set of voters in $V_{r+1}$ that voted for $c_{r+1}$. 
    Let $G_{r+1}$ denote the representatives of $c_{r+1}$.
    At most $n/k$ of the voters are chosen and added to $G_{r+1}$ (if there are more than $n/k$ voters in $H$, then exactly $n/k$ voters are chosen lexicographically). 
    $V_{r+2}$ is set like so: $V_{r+2} = V_{r+1} - G_{r+1}$.
    
    For the case where $n$ is not divisible by $k$, the upper bound used for the number of selected voters each round is set like so.
    Let $d = n \mod k$. For rounds $1$ to $d$, 
    $\left \lceil n/k \right \rceil $ voters
    are selected.
    For rounds $d+1$ to $k$
    $\left \lfloor n/k \right \rfloor $ voters
    are selected.
\end{definition}

\begin{theorem}
    Computing the winning committee for Greedy Monroe is inherently sequential, even when tie-breaking is never invoked.
\end{theorem}
\begin{proof}
    Suppose we are given an instance $(G, v)$ of LFMIS, where $G$ is a 3-regular graph with $m$ vertices labeled $1$ to $m$. We will construct an election with candidates $C$, collection of votes $V$, and committee size $k$, and specify a candidate $c \in C$ such that the winning committee selected by GM includes $c$ if and only if $v$ is included in the LFMIS of $G$.

    Let $C$ have a candidate for each vertex of $G$.
    For each $\{a,b\} \in E(G)$, we add $2m$ votes $\{a,b\}$ to $V$.
    We also add $(m-a)$ singleton votes for each $a \in C$.
    Note that each candidate $a \in \{1, 2, \ldots, m\}$ has $6m + (m-a)$ votes. 
    Finally, add $m$ extra candidates, labeled $m+1$ to $2m$, that have $5m+1, 5m+2, \ldots, 6m$ votes.
    It is clear that there are no ties in the election.
    Finally, we set $c = v$ and $k = m$.
    We now show that the candidates picked from $V(G)$ correspond exactly to the vertices in the LFMIS of $G$. Observe that 
    \[\frac{n}{k} = \frac{(\frac{3m}{2})2m + \sum_{a=1}^m (m-a) + \sum_{a=1}^m (5m+a)}{m} = 9m \]

    Note that the maximum number of voters a candidate has is $6m + m-1 \leq 9m$. Thus, whenever a candidate is elected, it is selected as the representative of all of its voters. We now prove the following claim.

    \begin{claim}
        Let $a$ and $b$ be unelected candidates in round $r+1$. If $b$ shares votes with a previously elected candidate and $a$ does not, then $a$ has more votes in $V_{r+1}$ than $b$.
    \end{claim}
    \begin{proof}
    Since a candidate is the representative of all of its voters and $b$ shares $2m$ candidates with a previously elected voter, $b$ has at most $6m + m-1 - 2m = 5m-1$ votes in $V_{r+1}$.
    Any candidate $a$ that shares no votes with a previously elected candidate has at least $5m+1$ votes.
    \end{proof}

    By the claim above and the fact that candidates with more votes (i.e., vertices with earlier lexicographic order) are preferred, it is clear that the candidates picked will correspond exactly to the vertices in the LFMIS of $G$.
    
    Once adding candidates from $V(G)$ is suboptimal, we will start picking the candidates labeled $m+1$ to $2m$ until the algorithm terminates. 
    
    Note that $k$ is large enough (at least the number of vertices in $G$) so that every candidate corresponding to a vertex in the LFMIS of $G$ will be picked. Similarly, the number of extra candidates is large enough so that we are not forced to pick any candidates corresponding to the vertices in $G$ after the LFMIS candidates have been picked.
    
    Therefore, $v$ is picked in the committee by Greedy Monroe if and only if $v$ is part of the LFMIS. The reduction is clearly in logspace.
\end{proof}

\subsection{Method of Equal Shares + Phragmén}
\label{sec:mesphrag}

\citet{lac-sko:b:abc-rules} define the Method of Equal Shares as a two-phase rule like so. Suppose we want a committee of size $k$. 
In the first phase, we select $k'$ candidates as described in Definition~\ref{def:mes}. If $k' < k$, then we run another ABC rule on the remaining candidates to select a committee of size $k - k'$. Lackner and Skowron suggest using seq-Phragmén for the second phase, as defined below. We refer to this rule as MES+seq-P.

\begin{definition}[MES+seq-P]
    Suppose we are given candidates $C$, collection of votes $V$, and committee size $k$.
    MES+seq-P runs in two phases.
    In the first phase, we run MES: we select $k'$ candidates as described in Definition~\ref{def:mes}.
    If $k' < k$, then we run seq-Phragmén on the remaining candidates to pick a committee of size $k-k'$, where the loads of each voter $i$ are initialized by setting $y_0(i) = -x_{k'}(i)$.
\end{definition}

Although we show in our main text that MES is hard, we want to show the hardness of the rules as defined by \cite{lac-sko:b:abc-rules}. Thus, for the sake of completeness, we show that MES+seq-P is \p-hard.
Particularly, we reduce LFMIS to MES+seq-P such that the committee picked by MES corresponds exactly to the LFMIS of the input graph.

\begin{theorem}
Computing the winning committee for MES+seq-P is inherently sequential.
\end{theorem}

\begin{proof}    
Suppose we are given an instance $(G, v)$ of LFMIS, where $G$ is a 3-regular graph with $m$ vertices labeled $1$ to $m$. We will construct an election with candidates $C$, collection of votes $V$, and committee size $k$, and specify a candidate $c \in C$ such that the winning committee selected by MES+seq-P includes $c$ if and only if $v$ is included in the LFMIS of $G$.

As in our previous proof, we add to $C$ each vertex in $G$ and add to $V$ a vote for each edge in $G$.
We will also add many copies of candidates and votes equivalent to a 13-vertex tree graph, denoted $T$: $T$ has a root vertex with four children, each of which has two children. 
We add $13m$ ``copies of $T$'' to our election system, i.e.,
we add $13m$ extra candidates to $C$, labeled $m+1$ to $14m$, and $12m$ votes to $V$ equivalent to the edges of $T$.
Let $k = m/2 + 4m$ and $c=v$.

We now show that the candidates picked from $V(G)$ by MES+seq-P correspond exactly to the vertices in the LFMIS of $G$. Specifically, these candidates are picked in the first phase, by MES.

Recall that initially we are running MES. We have $n = 3m/2 + 12m$ votes, so the initial budget for each voter is $k/n = (m/2 + 4m)/(3m/2+12m) = 1/3$.
Observe that the candidates corresponding to the root of $T$ have four voters each. So, initially, all of these roots can be elected by incurring a cost of $1/4$ from each of their voters. Since this is the minimum cost, $\rho_c$, MES will pick each root for the winning committee, leaving each root's voters with budget $1/3 - 1/4 = 1/12$. After $m$ roots have been elected, observe that no other candidate corresponding to a vertex from $T$ can be elected by MES; those in the second layer (children of the roots) have voters with budget $1/3 + 1/3 + 1/12 < 1$, and those in the last layer (leaf vertices of $T$) have single voters with budget $1/3$.
We are now in a situation similar to the one in the proof of Theorem~\ref{thm:mes-main}; the candidates corresponding to vertices in $V(G)$ can be elected by incurring a cost of $1/3$ on each of their voters. As argued in said proof, the candidates corresponding exactly to the LFMIS of $V(G)$ will be elected. Once all of these candidates have been elected, no candidates' voters have the budget to elect them and the first phase concludes.

Now we are in the second phase and running seq-Phragmén with initial loads $y_{0}(i) = -x_{k'}(i)$
for each voter $i$
where $x_{k'}(i)$ is the budget of voter $i$ at the end of MES's last round.
We now calculate the loads of electing every candidate not yet elected.
For every candidate corresponding to $V(G)$, we have 
\[
\ell_1(c) = \frac{1 + \left( -\frac{1}{3} - \frac{1}{3} \right)}{3} = \frac{1}{9}
\text{ or }
\ell_1(c) = \frac{1 + \left( -\frac{1}{3}\right)}{3} = \frac{2}{9}
\]
since it has either 1 or 2 voters with budget $1/3$ at the end of phase 1.
The candidates corresponding to vertices in the second layer of $T$ have loads
\[
\ell_1(c) = \frac{1 + \left( -\frac{1}{3} - \frac{1}{3} 
-\frac{1}{12} \right)}{3} = \frac{1}{12}
\]
The candidates corresponding to vertices in the last layer of $T$ have loads
\[
\ell_1(c) = \frac{1 + \left( -\frac{1}{3} \right)}{3} = \frac{2}{9}
\]

We need not consider the roots of $T$ since those were already elected in phase 1. Observe that the candidates corresponding to vertices in the second layer of $T$ have the lowest load. Moreover, note that each of the votes for these vertices are disjoint, since these vertices share no edges, i.e., we do not need to be concerned with the updated loads of each voter in the next round. Thus, all $4m$ of these vertices can be picked sequentially until we elected a total of $k$ candidates.

Note that $k$ is large enough (at least the number of vertices in $G$) so that every candidate corresponding to a vertex in the LFMIS of $G$ will be picked. Similarly, the number of extra candidates is large enough so that we are not forced to pick any candidates corresponding to the vertices in $G$ after the LFMIS candidates have been picked; there are $5m$ extra candidates that can be picked (five for each tree) and $m/2 + 4m = 4.5m < 5m$ candidates that need to be elected.

Therefore, $v$ is picked in the winning committee by MES+seq-P if and only if $v$ is part of the LFMIS. The reduction is clearly in logspace. \end{proof}

\subsection{MES without tie-breaking}
\label{sec:mes-noties}

\noindent
{\bf Theorem~\ref{mes:noties}.}
{\em Computing the winning committee for (the first phase of) MES is inherently sequential even when tie-breaking is never invoked.}

\begin{proof}
    Suppose we are given an instance $(G, v)$ of LFMIS, where $G$ is a 3-regular graph with $m$ vertices labeled $1$ to $m$. Without loss of generality, we will assume that $m \geq 12$.
    We will construct an election with candidate $C$, collection of votes $V$, and committee size $k$, and specify a candidate $c \in C$ such that the winning committee selected by MES includes $c$ if and only if $v$ is included in the LFMIS of $G$.

    Let $C$ have a candidate for each vertex of $G$.
    For each $\{a,b\} \in E(G)$, we add $m^2$ votes $\{a,b\}$ to $V$.
    We also add $(m-a)$ votes for each $a \in C$.
    Note that each candidate  $a$ has $3m^2 + (m-a)$ votes. 
    Finally, add an extra candidate, labeled 0, that has $m^3$ singleton votes.
    It is clear that there are no ties in the election.
    Note that we have a total of $n = 3m^3/2 + m(m+1)/2 + m^3$ votes.
    Finally, we set $c = v$ and $k = m+1$.
    Observe that the initial budget is
    \[x_0(i) = \frac{k}{n} = \frac{m+1}{3m^3/2 + m(m-1)/2 + m^3}\]
    for each voter $i$.
    We now show that the candidates corresponding to $V(G)$ picked by MES correspond exactly to the vertices in the LFMIS of $G$.
    We prove this statement via the following two claims for all candidates other than candidate 0.
    \begin{claim}
    \label{clm:mes-nt-1}
        For an unelected candidate $a>0$, if its voters have not
        already voted for another candidate, $a$'s voters have the budget to elect $a$.
    \end{claim}
    \begin{proof}
        A candidate can be elected in round $r$ if $\sum_{i \in N(a)} x_r(i) \geq 1$.
        If $a$'s voters have not voted for anyone else, then $x_r(i) = x_0(i)$. Thus, 
        for each candidate $a$, we have 
        \begin{align*}
        \sum_{i \in N(a)} x_0(i) &=  x_0(i) \times (3m^2 + (m-a) ) \\
        &\geq  x_0(m) \times (3m^2 + (m-m) ) \\
        &\geq \frac{3m^2(m+1)}{3m^3/2 + m(m-1)/2 + m^3}
        \end{align*} 
        The R.H.S. is $\geq 1$ for all $m \geq 1$.
    \end{proof}
    
    \begin{claim}
    \label{clm:mes-nt-2}
    For candidates $a > 0$ and $b > 0$ that share a vote, if $b$ was selected by the MES in a previous round, then the voters of $a$ do not have the budget to elect $a$.
    \end{claim}
    \begin{proof}

    Candidate $b$ has $3m^2 + (m-b)$ votes. The minimum amount that each voter spent to elect $b$ is $1/(3m^2 + (m-b))$. Of these votes,
    $m^2$ are shared with $a$.
    Thus, the maximum amount of total budget that the voters of $a$ have is
    \begin{align*}
        \sum_{i \in N(a)} x_r(i) &\leq 
         (2m^2 + (m-a)) x_0(i)   + \left( x_0(i) - \frac{1}{(3m^2 + (m-b))} \right) m^2 \\
         &\leq (2m^2 + m) x_0(i)   + \left( x_0(i) - \frac{1}{(3m^2 + m)} \right) m^2
    \end{align*}
    Substituting $x_0(i)$ in the R.H.S. and simplifying gives us the following expression:
    \[ \frac{2(17m+8)}{5(5m^2+m-1)} + \frac{1}{3(3m+1)} + \frac{13}{15} \]
    This is clearly decreasing as $m$ grows, and it can be verified to be $<1$ for all $ m \geq 12$.
    \end{proof}
    
    Recall that MES prefers picking the candidate that minimizes the maximum cost incurred by each voter. Candidate 0 is picked first since it has the most votes and the cost incurred is split among these voters.
    Similar reasoning shows that, candidate 1 is picked since it has the most voters ($3m^2 + m-1$).
    By Claims~\ref{clm:mes-nt-1} and~\ref{clm:mes-nt-2} and the fact that candidates with more votes (i.e., vertices with earlier lexicographic order) are preferred, it is clear that the candidates picked after candidate 1 will correspond exactly to the vertices in the LFMIS of $G$.
    Since we can pick up to $m-1$ remaining candidates, we are guaranteed to pick all the candidates corresponding to the LFMIS of $G$.
    
    Note that the extra padding of $m^3$ votes given by candidate 0 is what allows us to be in the ``goldilocks zone'' where the budget is large enough so that every candidate has large enough initial budget to be elected, 
    but taking away some of the budget of $m^2$ of its votes does not allow it to be elected.
    
    It is clear that $v \in V(G)$ is elected by MES if and only if $v$ is part of $G$'s LFMIS, and that the reduction is clearly in logspace, thus completing our proof.
    \end{proof}

\end{document}